\pdfoutput=1
\documentclass[conference]{IEEEtran}

\usepackage{physics}
\usepackage{cite}
\usepackage{amsmath,amssymb,amsfonts}
\usepackage[pdftex]{graphicx}
\graphicspath{{./fig3/}}
\usepackage{textcomp}
\usepackage{subfigure}
\usepackage{mathtools}

\usepackage{amsthm}

\newtheorem{theorem}{Theorem}[section]

\usepackage{fancybox}
\usepackage{xcolor}

\newcommand{\nouveautxt}[1]{#1}

\def\BibTeX{{\rm B\kern-.05em{\sc i\kern-.025em b}\kern-.08em
    T\kern-.1667em\lower.7ex\hbox{E}\kern-.125emX}}
\begin{document}
\ifpdf
\DeclareGraphicsExtensions{.pdf}
\else
\DeclareGraphicsExtensions{.eps}
\fi

%
\title{Superimposed Frame Synchronization Optimization \\for Finite Blocklength Regime}

\author{\IEEEauthorblockN{Alex The Phuong Nguyen}
\IEEEauthorblockA{\textit{IMT Atlantique, Lab-STICC, UBL} \\
29238 Brest, France \\
thephuong.nguyen@imt-atlantique.fr}
\and
\IEEEauthorblockN{Rapha\"el Le Bidan}
\IEEEauthorblockA{\textit{IMT Atlantique, Lab-STICC, UBL} \\
29238 Brest, France\\
raphael.lebidan@imt-atlantique.fr}
\and
\IEEEauthorblockN{Fr\'ed\'eric Guilloud}
\IEEEauthorblockA{\textit{IMT Atlantique, Lab-STICC, UBL} \\
29238 Brest, France\\
frederic.guilloud@imt-atlantique.fr}
}

\maketitle

\begin{abstract}
Considering a short frame length, which is typical in Ultra-Reliable Low-Latency and massive Machine Type Communications, a trade-off exists between improving the performance of frame synchronization (FS) and improving the performance of information throughput. In this paper, we consider the case of continuous transmission over AWGN channels where the synchronization sequence is superimposed to the data symbols, as opposed to being added as a frame header. The advantage of this superposition is that the synchronization length is as long as the frame length. On the other hand, its power has to be traded-off not to degrade the code performance. We first provide the analysis of FS error probability using an approximation of the probability distribution of the overall received signal. Numerical evaluations show the tightness of our analytic results. Then we optimize the fraction of power allocated to the superimposed synchronization sequence in order to maximize the probability of receiving a frame without synchronization errors nor decoding errors. Comparison of the theoretical model predictions to a practical setup show very close optimal power allocation policies. 
\end{abstract}


\section{Introduction}
Short packet transmissions are fundamental bricks to build future applications such as intelligent houses, connected vehicles, automated factories and Internet of Things. The physical layer design rules for short packet transmissions are quite different from the classic asymptotic counterpart where blocklengths are assumed asymptotically long. For example, inserting pilots for channel estimation becomes surprisingly expensive~\cite{ostman2017finite}, capacity achieving codes like LDPC and turbo code do not perform as well as their lengths are reduced~\cite{3gppfblcode2016} and the widely-used ergodic and outage capacities are overshooting~\cite{durisi2016short}. As a result, new finite blocklength (FBL) design rules are required.
Recently, the bounds of maximal channel coding rate for the FBL regime, which are the peers of asymptotic capacity metrics, were derived in~\cite{polyanskiy2010channel,erseghe2016coding}. Nevertheless, successful frame synchronization (FS) is also required prior to decoding to achieve the promised throughput. FS can be achieved by blind estimators~\cite{imad2009blind}. However, for complexity issues, the use of sequences of known symbols denoted synchronization words (SW) are preferred. SW can be concatenated to the information symbols by means of a frame header~\cite{chiani2007analysis}. When the length of frame is fixed, including a header for FS reduces both spectral efficiency and the coding length. Conversely, reducing the FS length will lower the FS performance. Hence there exists a trade-off to be found to optimize the chance of receiving a frame without errors. In the context of low latency communications and/or massive connectivity, the frame length can be reduced while keeping a maximal FS length (i.e. the length of the entire frame) by superimposing SW to data symbols~\cite{wang2007novel}. Then the optimization is related to the power of the FS symbols.

FS is a well investigated topic. FS methods can be roughly divided into two groups according to two distinct transmission modes. The first one is \emph{burst transmissions} for which FS is performed thanks to binary hypothesis testing. Binary hypothesis testing requires a threshold that is usually derived from the Neyman-Pearson lemma. Optimum metrics, in terms of minimizing FS error probability, for binary modulation over AWGN channels are studied in~\cite{chiani2005practical,chiani2006sequential} and~\cite{suwansantisuk2008frame}. The approximation of optimum metric for M-PSK modulation with phase offset error over AWGN channels is provided by~\cite{elzanaty2017frame}. In the finite blocklength regime, a study of FS for burst transmission was recently proposed in~\cite{bana2018short}.

The second mode is \emph{continuous transmissions} where frames are transmitted successively. The knowledge of frame length enables Maximum Likelihood (ML) synchronization: a metric to be maximized is computed for all possible SW positions. The synchronization metric to be used depends on channel models and data distributions. For example, the correlation metric minimizes the FS error probability for binary symmetric channels (BSC)~\cite{barker1953group}, whereas for AWGN channels with binary modulation, the optimal metric requires an additional energy-correction term~\cite{massey1972optimum}. As an extension of~\cite{massey1972optimum}, the authors of~\cite{lui1987frame} present ML metrics for general M-ary phase-coherent and phase-noncoherent AWGN channels. For Rayleigh fading models, the optimum metrics for non-coherent detection is provided in~\cite{jia2005frame}. Regarding the analytic performance (the probability of erroneous frame synchronization), little work has been proposed so far. The work from~\cite{chiani2007analysis} is one of a few that we are aware of, and provides both metrics and analytic performance results for the particular case of binary modulations over AWGN channels.

In this paper, we provide an analysis of FS error probability for continuous transmissions with superimposed SW over AWGN channels. 
In the considered setup FS is performed once for the whole transmission rather than repeated on each frame. 
We will assume hereafter that the coded data symbols follow a spherical uniform distribution. Resorting to Gaussian approximation leads us to a simple expression that can be used to optimize the trade-off between FS overhead and FBL decoding by using the results of~\cite{erseghe2016coding}. Numerical evaluations show the tightness of our analytic results. Finally, this theoretical model is compared to a practical setup to show that their optimal power allocations are very close.
A similar issue has been tackled in~\cite{bana2018short} but with a different setup. Indeed, the authors of~\cite{bana2018short} assume \emph{burst transmissions} thus implementing binary hypothesis testing detection. On the contrary, our setup assumes \emph{continuous transmissions} for which the derivation of optimal thresholds is not required,
leading to different and complementary calculations and results.

The paper is organized as follows.
The system model and the optimization target are described in Section~\ref{sect:system_model}.
The analytic FS error probability is derived in Section~\ref{sect:fs_analysis} where an approximation of the received data distribution is proposed and justified.
Applications to the optimization of the synchronization word power, numerical evaluations and comparisons to a practical setup are reported in Section~\ref{sect:numerical_eval}.
Finally, conclusions and perspectives can be found in Section~\ref{sect:conclusion}.

Throughout the paper, random variables are denoted by upper case letters
 whereas their realizations are written in lower case, and boldface letters denote column-wise vectors and matrices. Hence, $X$ and $x$ stand for a random scalar variable and its realization respectively, whereas $\mathbf{X}$ and $\mathbf{x}$ are used to denote a random vector and its deterministic version.
Notation $[\mathbf{a};\mathbf{b}]$ denotes the vertical concatenation of two vectors $\mathbf{a}$ and $\mathbf{b}$.
Superscripts $^T$ and $^H$ stand for transpose and Hermitian transpose, respectively.
$\norm{.}$ is used to denote the $L_2$-norm and $\Re{x}$ is the real part of complex $x$.
We adopt conventional notations $\mathbf{0}$ and $\mathbf{I}$ for a zero matrix and an identity matrix of whatever dimension needed to make sense of the expressions being computed. To be more explicit, $\mathbf{0}_n$ and $\mathbf{I}_n$ can be used to clarify the dimensions $n \times n$ of these particular matrices.
The Delta function, the Gamma function, the tail distribution function of the standard Normal distribution and the $n$-order modified Bessel function of the first kind are denoted by $\delta(.)$, $\Gamma(.)$, $\mathcal{Q}(.)$ and $\mathcal{I}_n(.)$, respectively.
\section{System model and Target}
\label{sect:system_model}

\subsection{Frame structure}
\begin{figure}
	\centering
	\includegraphics[width=\linewidth]{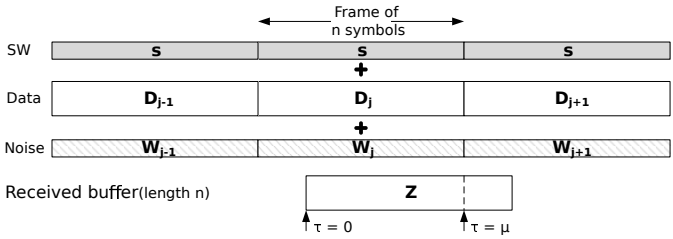}
	\caption{System model (from the $(j-1)$-th frame to the $(j+1)$-th frame). The frame beginning is at position $\tau=\mu$ in buffer $\mathbf{Z}$.}
	\label{fig:system_model}
\end{figure}

The frame is made of coded data $\mathbf{D} \in \mathbb{C}^n$ superimposed with a fixed known synchronization word (SW) $\mathbf{s} \in \mathbb{C}^n$ where $n$ is the number of symbol (channel-use) per frame.
Frames are assumed to be transmitted continuously over an AWGN channel. The receiver is synchronized at the symbol level. By buffering $n$ successive symbols, the frame beginning is surely included in the buffer which is denoted by $\mathbf{Z}$ (see \figurename~\ref{fig:system_model}). \nouveautxt{Note that although each frame includes a SW, the receiver only needs to synchronize once when starting the reception of the data stream. So the analysis hereafter is based on the probability of failure of this initial  synchronization step.}

The coded data $\mathbf{D}$ is assumed to follow the spherical uniform probability distribution with radius $\sqrt{n\rho}$: codewords $\mathbf{D} \in \mathbb{C}^n$ are uniformly distributed on the surface of a complex sphere and $\norm{\mathbf{D}}^2 = n\rho$. 
This assumption follows from Shannon's achievability proof of the AWGN channel capacity theorem that optimal codes for complex AWGN channels in the asymptotic regime consist of dense packing of signal points within a sphere.

The power of noise $\mathbf{W}$ is normalized to unity and the average energy per symbol of $\mathbf{s}$, $\mathbf{D}$ and of the superposition of the two over the entire frame are denoted by $\rho_s$, $\rho$ and $\rho_{\mbox{\scriptsize tot}}$ respectively. Hence $\norm{\mathbf{s}}^2 = n\rho_s$ with $\rho_{\mbox{\scriptsize tot}}=\rho_s+\rho$ and we define the \emph{overhead} ratio between the average symbol energy of SW and the average symbol energy of the entire frame as
\begin{align}
\alpha \triangleq \frac{\rho_s}{\rho_{\mbox{\scriptsize tot}}}. 
\end{align}
Intuitively, if $\alpha$ is increased, the frame beginning will be easier to locate but less power for data will lower the coding performance, and vice versa.

As illustrated in \figurename~\ref{fig:system_model}, if frame beginning is at position $\mu$ of buffer $\mathbf{Z}$, then $\mathbf{Z} = \mathbf{s}_\mu + \mathbf{D}_z + \mathbf{W}_z$ where $\mathbf{s}_\mu$ denotes the right circular shift by $\mu$ positions of $\mathbf{s}$, $\mathbf{D}_z$ is the concatenation of a part of frame-$(j)$-data $\mathbf{D}_j$ and a part of frame-$(j+1)$-data $\mathbf{D}_{j+1}$, $\mathbf{W}_z$ is an i.i.d. Gaussian noise vector. Hence $\mathbf{Z}$ consists of the deterministic $\mathbf{s}_\mu$ and the random part $\mathbf{Y}_z \triangleq \mathbf{D}_z + \mathbf{W}_z$.

The frame synchronization can be modeled as the optimization problem
$\hat{\tau} = \underset{0\leq \tau < n}{\textrm{argmax }} f(\mathbf{z},\tau)$
where $\tau$ is the tested position of the frame beginning within the received vector $\mathbf{z}$ and $f(\mathbf{z},\tau)$ is a \emph{decision metric}.
If the SW is at position $0 \leq \mu < n$, with $f(\mathbf{Z},\tau)$ denoting \emph{decision-metric random variables}, the probability $P_e$ of erroneous synchronization can be written as
$P_e = \textrm{Pr} \left\lbrace \hat{\tau} \neq \mu \right\rbrace = \textrm{Pr} \left\lbrace \bigcup\limits_{\tau \neq \mu} f(\mathbf{Z},\mu) \leq f(\mathbf{Z},\tau) \right\rbrace$
which can be upper-bounded by
\begin{equation}
\label{eq:union_bound}
P_e \leq P_{e,u} \triangleq \sum_{\tau \neq \mu} \mathcal{P}_e(\tau),
\end{equation}
where $\mathcal{P}_e(\tau) \triangleq \textrm{Pr} \left\lbrace f(\mathbf{Z},\mu) \leq f(\mathbf{Z},\tau) \right\rbrace$.
In this paper, we focus on the widely-used
correlation rule $f(\mathbf{z},\tau) = \Re{\mathbf{s}_\tau^H \mathbf{z}}$~\cite[Chapter 3, Page 69]{ling:2017} where the channel phase is assumed to be well estimated and compensated.

\subsection{Optimization target}

The objective is to find the optimum value of the energy allocation ratio (the overhead) $\alpha$ for a given fixed total symbol energy $\rho_{\mbox{\scriptsize tot}}$. We optimize $\alpha$ so as to minimize the frame error rate $P_f$ defined by
\begin{align}
\label{eq:fer_model}
    P_f \triangleq 1-(1-P_e)(1-\epsilon^\star)
\end{align}
where $P_e$ is the probability of erroneous synchronization and where $\epsilon^\star$ denotes the minimal channel-coding packet error probability that one can obtain with the best channel code/decoder pair for a finite-length-$n$ packet of $k$ information bits at SNR $\rho$~\cite{polyanskiy2010channel}. Rigorously speaking, $P_f$ is an upper-bound on the true frame error rate since it disregards the possibility that, at least in principle, a frame could be correctly decoded even though it has not been correctly synchronized.

Under the spherical uniform distribution assumption, $\epsilon^\star$ can be approximated by
\begin{align}
\label{eq:epsilon_D}
\epsilon^\star(n,k,\rho) \approx \mathcal{Q}\left(\frac{nC(\rho) - k+\frac{1}{2}\log_2(2n)}{\sqrt{nV(\rho)}}\right)    
\end{align}
where $C(\rho) = \log_2(1+\rho)$ and $V(\rho) = \log_2^2(e)\frac{\rho(\rho+2)}{(\rho+1)^2}$ for the \nouveautxt{\emph{complex}} continuous-input AWGN channel~\cite{erseghe2016coding}. \nouveautxt{One could alternatively use any other achievability bounds of \cite{polyanskiy2010channel}, e.g. RCU and $\kappa \beta$, to bound $\epsilon^\star$ in \eqref{eq:fer_model}.}

Replacing $P_e$ by its union bound~\eqref{eq:union_bound} results in the following upper bound $P_{f,u}$ on frame error rate
\begin{align}
\label{eq:FBL_optim_Pf}
P_{f,u} \triangleq 1-(1-P_{e,u})(1-\epsilon^\star).
\end{align}
The optimization target becomes finding the overhead $\alpha$ that minimizes $P_{f,u}$:
\begin{align}
    \hat{\alpha} = \textrm{argmin}_{0 \le \alpha \le 1} P_{f,u}.
\end{align}
The evaluation of $P_{f,u}$ is the subject of the next Section.
\section{FS error probability analysis}
\label{sect:fs_analysis}
As in~\eqref{eq:union_bound}, the distributions of $\mathbf{Z}$ and $\mathbf{Y}_z$ are required for FS error probability analysis. Since the distribution of $\mathbf{Y}_z$ is quite involved, we propose hereafter an approximation of this distribution which will be used further to obtain an analytic expression of FS error probability.

In Section~\ref{sect:Y_pdf}, we first analyze the distribution of $\mathbf{Y} = \mathbf{D} + \mathbf{W}$ (i.e. assuming $\mu=0$) where $\mathbf{D} \in \mathbb{R}^n$ follows a uniform distribution on the surface of a sphere of radius $\sqrt{n\rho}$ and $\mathbf{W} \sim \mathcal{N}(\mathbf{0},\mathbf{I}_n)$.
An approximation of the distribution of $\mathbf{Y}$ that simplifies the FS analysis is given in Section~\ref{sect:approx_KL_optim_real}.
The choice of the real space is simply for illustrative purpose. Section~\ref{sect:complex_space_result} provides the extension to the complex space that is used to obtain final results on FS error probability in Section~\ref{sect:FS_result}. 

\subsection{The probability density function (PDF) of $\mathbf{Y}$}
\label{sect:Y_pdf}

\begin{theorem}[Distribution of $\mathbf{Y}$ for real space]
\label{Theorem:dist_real_Yd}
Let $\mathbf{D} \in \mathbb{R}^n$ follow uniform distribution on the surface of a sphere of radius $\sqrt{n\rho}$ and $\mathbf{W} \sim \mathcal{N}(\mathbf{0},\mathbf{I}_n)$, then $\mathbf{Y} = \mathbf{D}+\mathbf{W}$ has PDF
\begin{align}
\begin{split}
p_{\mathbf{Y}}(\mathbf{y}) = &\frac{\Gamma(n/2)}{2\pi^{n/2}} (n\rho)^{1/2-n/4} \norm{\mathbf{y}}^{1-n/2} e^{-(\norm{\mathbf{y}}^2 + n\rho)/2} \\
&\times \mathcal{I}_{n/2-1}(\norm{\mathbf{y}}\sqrt{n\rho}).
\end{split}
\end{align}
\end{theorem}

\begin{proof}[Proof]
We have $\norm{\mathbf{Y}}^2 = n\rho + \norm{\mathbf{W}}^2 + \mathbf{D}^H \mathbf{W} + \mathbf{W}^H \mathbf{D}$, then $p_{\mathbf{Y}}(\mathbf{y}) = p_{\mathbf{D},\mathbf{W}}(\mathbf{d},\mathbf{w})$.
Because of radial symmetry, the PDF of $\norm{\mathbf{Y}}^2$ is independent of the rotation applied on $\mathbf{D}$.
Indeed, by applying any unitary matrix $\mathbf{v} \in \mathbb{R}^{n\times n}$ on the data $\mathbf{D}$, the corresponding received random vector has the following square norm: $\norm{\mathbf{vD}+\mathbf{W}}^2 = n\rho + \norm{\mathbf{W}}^2 + \mathbf{D}^T \mathbf{v}^T \mathbf{W} + \mathbf{W}^T \mathbf{v}\mathbf{D}$. Then $p_{\mathbf{D},\mathbf{W}}(\mathbf{v}\mathbf{d},\mathbf{w}) = p_{\mathbf{D},\mathbf{W}}(\mathbf{d},\mathbf{v}^T \mathbf{w}) = p_{\mathbf{D},\mathbf{W}}(\mathbf{d},\mathbf{w})$. As a consequence, we can choose \nouveautxt{$\mathbf{d} = \mathbf{d}_0 \triangleq [\sqrt{n\rho};0;...;0]$} to show that $\norm{\mathbf{Y}}^2$ is distributed according to the non-central $\chi^2$ PDF with $n$ degrees of freedom and a non-centrality parameter of $n\rho$, $\norm{\mathbf{Y}}^2 \sim \chi_{n}^2(n\rho)$. Hence $\norm{\mathbf{Y}}$ follows a non-central Chi distribution $\chi_{n}(\sqrt{n\rho})$.

Again, the radial symmetry implies that $\mathbf{Y}$ is statistically invariant against rotation. Thus for a given $\norm{\mathbf{Y}}=r$, vector $\mathbf{Y}$ is uniformly distributed on the $n$-dimensional hyper-sphere of surface $\mathcal{S}(r) = \frac{2\pi^{n/2}}{\Gamma(n/2)} r^{n-1}$.

By a change of variables from Cartesian coordinates $\mathbf{y}$ to spherical coordinates $[r;\phi_1;...;\phi_{n-1}]$ for which the determinant of the Jacobian matrix is $r^{n-1} \prod_{j=1}^{n-2} \sin^{n-j-1} (\phi_j)$, the conditional PDF $p_{\mathbf{Y}|R=r}(\mathbf{y}|r)=\frac{1}{\mathcal{S}(r)} \delta(\norm{\mathbf{y}}=r)$ is obtained. 
Finally, the PDF of $\mathbf{Y}$ is obtained through $p_{\mathbf{Y}}(\mathbf{y}) = \int_r p_{\mathbf{Y}|R=r}(\mathbf{y}|r) p_R(r) \mathrm{d}r = \frac{1}{\mathcal{S}(\norm{\mathbf{y}})} p_{\norm{\mathbf{Y}}}(\norm{\mathbf{y}})$.
\end{proof}

The result of Theorem~\ref{Theorem:dist_real_Yd} is illustrated in \figurename~\ref{fig:visual_2_dim}.
\begin{figure}
	\centering
	\includegraphics[width=0.8\linewidth]{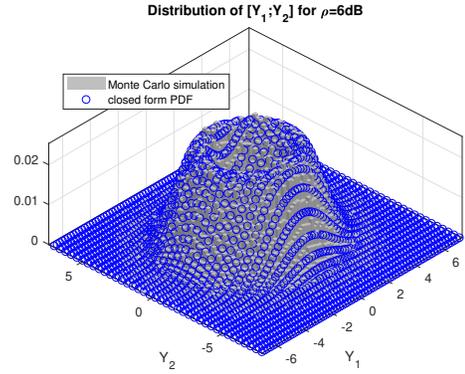}
	\caption{PDF of $\mathbf{Y} \in \mathbb{R}^2$ at $\rho=6\textrm{dB}$.}
	\label{fig:visual_2_dim}
\end{figure}

\subsection{Gaussian approximation of the PDF of $\mathbf{Y}$}
\label{sect:approx_KL_optim_real}
Although the PDF from Theorem~\ref{Theorem:dist_real_Yd} is numerically computable, it is still very cumbersome.
We propose hereafter a Gaussian approximation of $p_{\mathbf{Y}}(\mathbf{y})$.

\begin{theorem}[Gaussian approximation of $\mathbf{Y}$ for real space]
	\label{Theorem:approx_dist_real_Yd}
	Let $\mathbf{D} \in \mathbb{R}^n$ follow a uniform distribution on the surface of a sphere of radius $\sqrt{n\rho}$ and $\mathbf{W} \sim \mathcal{N}(\mathbf{0},\mathbf{I}_n)$. Then among all the centered Gaussian distribution, the one with covariance matrix ${\Sigma}=(1+\rho) \mathbf{I}_n$ minimizes its Kullback-Leibler divergence to the distribution of $\mathbf{Y} = \mathbf{D}+\mathbf{W}$.
\end{theorem}

\begin{proof}[Proof]
 Let $\mathbf{Q}\sim \mathcal{N}(\mathbf{0},{\Sigma})$ denote a zero-mean Gaussian random vector with a $n\times n$ covariance matrix ${\Sigma}$. The PDF of $\mathbf{Q}$ is given by
\begin{align*}
p_{\mathbf{Q}}(\mathbf{q}) = \frac{1}{(2\pi)^{n/2} \sqrt{\det({\Sigma})}} e^{-\frac{1}{2}\mathbf{q}^T {\Sigma}^{-1} \mathbf{q}}.
\end{align*}
The covariance matrix ${\Sigma}$ can be obtained by minimizing the Kullback-Leibler (KL) divergence between the distributions represented by $p_{\mathbf{Q}}$ and $p_{\mathbf{Y}}$:
\begin{align}
\label{eq:KL_real_q_correlated}
\begin{split}
&D(p_{\mathbf{Y}}||p_{\mathbf{Q}}) = \mathbb{E} \left[ \log(\frac{p_{\mathbf{Y}}(\mathbf{Y})}{p_{\mathbf{Q}}(\mathbf{Y})}) \right]
= \log\Gamma(n/2)\\
&+ \left(\frac{n}{2} - 1 \right) \log 2  + \left(\frac{1}{2} - \frac{n}{4}\right) \log(n\rho) - \frac{1}{2} n\rho\\
&+ \left(1-\frac{n}{2}\right) \mathbb{E} \left[ \log \norm{\mathbf{Y}} \right]
+ \mathbb{E} \left[ \log \mathcal{I}_{n/2-1}(\norm{\mathbf{Y}}\sqrt{n\rho}) \right]\\
&+ \frac{1}{2}\log \det({\Sigma}) + \frac{1}{2} \mathbb{E} \left[ \mathbf{Y}^T ({\Sigma}^{-1} - \mathbf{I}_n) \mathbf{Y} \right].
\end{split}
\end{align}

Let ${\Sigma}^*$ denote the optimal covariance matrix given by the minimization of~\eqref{eq:KL_real_q_correlated},
${\Sigma}^* \triangleq \underset{{\Sigma}}{\textrm{argmin }} D(p_{\mathbf{Y}}||p_{\mathbf{Q}}) = \underset{{\Sigma}}{\textrm{argmin }} \log \det \Sigma + \mathbb{E} \left[ \mathbf{Y}^T {\Sigma}^{-1} \mathbf{Y} \right]$.
Taking the derivative of $D(p_{\mathbf{Y}}||p_{\mathbf{Q}})$ with respect to $\tilde{{\Sigma}} \triangleq {\Sigma}^{-1}$ yields:
\begin{align}
\begin{split}
\frac{\partial D(p_{\mathbf{Y}}||p_{\mathbf{Q}})}{\partial \tilde{{\Sigma}}} &= 
\frac{\partial}{\partial \tilde{{\Sigma}}} \left(  \log \det \Sigma + \mathbb{E} \left[ \mathbf{Y}^T {\Sigma}^{-1} \mathbf{Y} \right] \right)\\
&= \frac{\partial}{\partial \tilde{{\Sigma}}} \left(  -\log \det \tilde{{\Sigma}} + \mathbb{E} \left[ \mathbf{Y}^T \tilde{{\Sigma}} \mathbf{Y} \right] \right).
\end{split}
\end{align}
Thanks to the Jacobi formula of matrix calculus which expresses the derivative of the determinant we obtain $\frac{\partial}{\partial \tilde{{\Sigma}}} \log \det \tilde{{\Sigma}} = {\Sigma}^T$. Also, $\frac{\partial}{\partial \tilde{{\Sigma}}} \mathbb{E} \left[ \mathbf{Y}^T \tilde{{\Sigma}} \mathbf{Y} \right] = \mathbb{E} \left[ \mathbf{Y} \mathbf{Y}^T \right]$. Hence the solution of $\frac{\partial D(p_{\mathbf{Y}}||p_{\mathbf{Q}})}{\partial \tilde{{\Sigma}}} = 0$ is given by:
\begin{align}
\label{eq:sigma_opt}
{\Sigma}^* = \mathbb{E} \left[ \mathbf{Y} \mathbf{Y}^T \right].
\end{align}

We recall that $\mathbf{Y} = \mathbf{D} + \mathbf{W}$ where the Cartesian coordinates of Gaussian noise $\mathbf{W}$ are identically distributed uncorrelated and have zero mean by definition. The uniform spherical distribution of $\mathbf{D}$ implies the same properties because two $n$-dimensional random vectors $[D_1;...;+D_l;...;D_n]$ and $[D_1;...;-D_l;...;D_n]$ follow the same distributions hence they must have the same mean, i.e. $\mathbb{E}[D_l] = \mathbb{E}[-D_l]$, and the same covariance, i.e. $\mathbb{E}[D_l D_m] = \mathbb{E}[-D_l D_m]$, $\forall m \neq l$.
Therefore the Cartesian coordinates of $\mathbf{D}$ are also zero-mean identically distributed uncorrelated random variables. As a consequence, the Cartesian coordinates of $\mathbf{Y}$ share the same properties and~\eqref{eq:sigma_opt} becomes
${\Sigma}^* = \mathbb{E}\left[ \mathbf{Y} \mathbf{Y}^T \right] = \sigma^2 \mathbf{I}_n$
where $\sigma^2$ denotes the variance of Cartesian coordinates of $\mathbf{Y}$.

The value of $\sigma^2$ is derived as follows.
By using the non-central law $\chi_{n}^2(n\rho)$ of $\norm{\mathbf{Y}}^2$ (see proof of Theorem~\ref{Theorem:dist_real_Yd}), $\mathbb{E} \left[ \norm{\mathbf{Y}}^2 \right] = n+n\rho$. On the other hand, since the Cartesian coordinates of $\mathbf{Y}$ are identically distributed, we have $\mathbb{E} \left[ \norm{\mathbf{Y}}^2 \right] = n \sigma^2$. Hence $\sigma^2 = 1 + \rho$.
Finally, we conclude that among centered Normal distributions $\mathcal{N}(\mathbf{0},{\Sigma})$, the one minimizing the KL divergence is the i.i.d. Normal distribution $\mathcal{N}(\mathbf{0},(1+\rho) \mathbf{I}_n)$.
\end{proof}

The KL divergence between the distribution of $\mathbf{Y}$ and its Gaussian approximation is illustrated in \figurename~\ref{fig:KL_Qu_optim} for several SNR $\rho$, showing that it is almost independent of $n$. This means that the approximation is good even for small $n$.
Also, it is observed that the lower the SNR, the better the approximation. Nevertheless, for a wide range of SNR, the divergence remains small enough to enable a reliable approximation of the theoretical probability of synchronization error as detailed hereafter.
\begin{figure}
	\centering
	\includegraphics[width=0.8\linewidth]{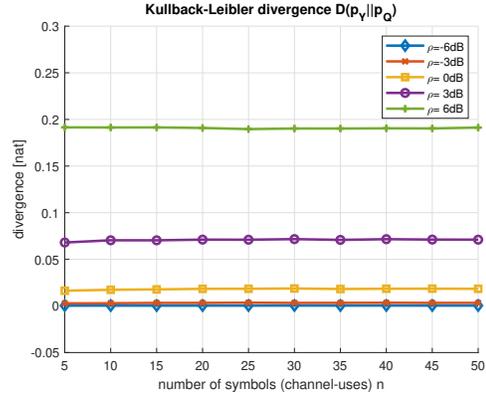}
	\caption{KL divergences $D(p_{\mathbf{Y}}||p_{\mathbf{Q}})$ for several SNR $\rho$ as a function of the number of symbols/channel-uses.}
	\label{fig:KL_Qu_optim}
\end{figure}

\subsection{Extension to Complex Signals}
\label{sect:complex_space_result}
To analyze the FS error probability of the system model in Section~\ref{sect:system_model}, we provide hereafter the complex space versions of Theorems~\ref{Theorem:dist_real_Yd} and~\ref{Theorem:approx_dist_real_Yd}.

\begin{theorem}[Distribution of $\mathbf{Y}$ for complex space and its Gaussian approximation]
	\label{Theorem:dist_cpx_Yd}
	Let $\mathbf{D} \in \mathbb{C}^n$ follow uniform distribution on the surface of sphere having radius $\sqrt{n\rho}$ and $\mathbf{W} \sim \mathcal{CN}(\mathbf{0},\mathbf{I}_n)$, then $\mathbf{Y} = \mathbf{D}+\mathbf{W}$ has PDF
 	\begin{align}
 	\begin{split}
 	p_{\mathbf{Y}}(\mathbf{y}) = &\Gamma(n)\pi^{-n} (n\rho)^{(1-n)/2} \norm{\mathbf{y}}^{1-n} e^{-(\norm{\mathbf{y}}^2 + n\rho)}\\
 	&\times \mathcal{I}_{n-1}(2\norm{\mathbf{y}}\sqrt{n\rho}).
 	\end{split}
 	\end{align}
Furthermore, among circular Normal distributions $Q=\mathcal{CN}({0},{\Sigma})$, the one minimizing the KL divergence $D(p_\mathbf{Y}||p_\mathbf{Q})$ is $\mathcal{CN}(\mathbf{0},(1+\rho) \mathbf{I}_n)$.
\end{theorem}
\begin{proof}[Proof]
	The proof is similar to the one of Theorem~\ref{Theorem:dist_real_Yd} with slight differences in power normalization, and with the surface of \emph{complex} hyper-sphere having radius $r$ given by $\frac{2\pi^{n}}{\Gamma(n)} r^{2n-1}$. The optimal Gaussian approximation is obtained by the same reasoning than in  Section~\ref{sect:approx_KL_optim_real}.
\end{proof}

\subsection{FS error probability}
\label{sect:FS_result}
One advantage of the proposed distribution approximation is that by approximating both $(\mathbf{D}_j + \mathbf{W}_j)$ and $(\mathbf{D}_{j+1} + \mathbf{W}_{j+1})$ by $\mathcal{CN}(\mathbf{0},(1+\rho)\mathbf{I}_n)$, $\mathbf{Y}_z = (\mathbf{D}_z + \mathbf{W}_z)$ can also be approximated by $\mathcal{CN}(\mathbf{0},(1+\rho)\mathbf{I}_n)$. 

If the frame starts at position $\mu$ of buffer $\mathbf{Z}$, the correlation metric yields $f(\mathbf{Z},\mu) = \norm{\mathbf{s}}^2 + \Re{\mathbf{s}_\mu^H(\mathbf{D}_z + \mathbf{W}_z)}$ and $f(\mathbf{Z},\tau) = \Re{\mathbf{s}_\tau^H \mathbf{s}_\mu} + \Re{\mathbf{s}_\tau^H(\mathbf{D}_z + \mathbf{W}_z)}$.
Because of the Gaussian approximation of $\mathbf{Y}_z$, we obtain that random variable $(f(\mathbf{Z},\mu) - f_C(\mathbf{Z},\tau))$ can be well approximated by $\mathcal{CN}(\norm{\mathbf{s}}^2-\Re{\mathbf{s}_\tau^H \mathbf{s}_\mu},\frac{1+\rho}{2} \norm{\mathbf{s}_\mu - \mathbf{s}_\tau}^2)$.
Finally, 
\begin{align}
\label{eq:Pe_tau}
\begin{split}
\mathcal{P}_e(\tau) \!=\! \textrm{Pr} \left\lbrace f(\mathbf{Z},\mu) \!\!\leq\!\! f(\mathbf{Z},\tau) \right\rbrace
\!\approx\! \mathcal{Q} \left( \frac{\norm{\mathbf{s}}^2\!-\!\Re{\mathbf{s}_\tau^H \mathbf{s}_\mu}}{\sqrt{\frac{1+\rho}{2} \norm{\mathbf{s}_\mu \!-\! \mathbf{s}_\tau}^2}} \right).
\end{split}
\end{align}

For a given frame length $n$ and FS power allocation $\alpha$, the upper bound $P_{f,u}$ on frame error rate can be calculated by combining~\eqref{eq:Pe_tau} and~\eqref{eq:epsilon_D} into~\eqref{eq:FBL_optim_Pf}:
\begin{align}
\label{eq:Pfu}
\begin{split}
&P_{f,u} \approx 1 - 
\left( 1-\sum_{\tau \neq \mu} \mathcal{Q} \left( \frac{\norm{\mathbf{s}}^2-\Re{\mathbf{s}_\tau^H \mathbf{s}_\mu}}{\sqrt{\frac{1+\rho}{2} \norm{\mathbf{s}_\mu - \mathbf{s}_\tau}^2}} \right) \right)\\
& \times \left( 1-\mathcal{Q}\left(\frac{nC(\rho) - k+\frac{1}{2}\log_2(2n)}{\sqrt{nV(\rho)}}\right) \right).
\end{split}
\end{align}

Note that in~\eqref{eq:Pfu}, the overhead $\alpha$ is not explicit because $P_{f,u}$ depends on the specific structure of the SW $\mathbf{s}$. We can bring out $\alpha$ by using particular sequences that have ideal auto-correlation property $\Re{\mathbf{s}_\tau^H \mathbf{s}_\mu} = \delta(\tau-\mu)$ (e.g. Zadoff-Chu sequences~\cite{zepernick2013pseudo,3gpp362112017}). Then~\eqref{eq:Pe_tau} simplifies into
$\mathcal{P}_e(\tau) \approx \mathcal{Q}\left(\frac{\norm{\mathbf{s}}}{\sqrt{1+\rho}}\right) = \mathcal{Q}\left(\sqrt{n \frac{\rho_s}{1+\rho}}\right)$
and~\eqref{eq:Pfu} becomes
\begin{align}
\label{eq:Pfu_ZC}
\begin{split}
&P_{f,u} \approx 1 - 
\left( 1-(n-1) \mathcal{Q}\left(\sqrt{n \frac{\alpha \rho_{\mbox{\scriptsize tot}} }{1+(1-\alpha)\rho_{\mbox{\scriptsize tot}} }}\right)\right)\\
& \times \left( 1-\mathcal{Q}\left(\frac{nC((1-\alpha)\rho_{\mbox{\scriptsize tot}}) - k+\frac{1}{2}\log_2(2n)}{\sqrt{nV((1-\alpha)\rho_{\mbox{\scriptsize tot}})}}\right) \right).
\end{split}
\end{align}
\section{Numerical evaluation}
\label{sect:numerical_eval}
The upper bound on the probability of erroneous synchronization~\eqref{eq:union_bound} is evaluated based on ~\eqref{eq:Pe_tau} 
 and compared to Monte Carlo simulations. Zadoff-Chu sequences of root 1 are used as SW $\mathbf{s}$. The comparison is illustrated in~\figurename{}~\ref{fig:fs_Q_optim_2frame} for frame with short lengths. Despite the Gaussian assumption and the union bound approximation, the proposed upper bounds are tight compared to their corresponding Monte-Carlo simulations, even at high SNR.
It is worth emphasizing that while Monte-Carlo simulations remain manageable at SNRs such that $P_e > 10^{-8}$, the proposed theoretical approximation allows much faster evaluations for every SNR.

Thanks to the distribution approximation of the erroneous FS probability,
 the overhead ratio $\alpha$ that characterizes the fraction of total power allocated to the SW can be optimized. In \figurename{}~\ref{fig:fs_ed_32bits}, the optimization of the frame error rate under both synchronization errors and decoding errors is illustrated for several transmitted average symbol energy $\rho_{\mbox{\scriptsize tot}}$ for a short frame of $n=63$ symbols \nouveautxt{(the frame length is constrained by Zadoff-Chu sequences and should therefore be odd).}
 and $k=32$ information bits.
 The small mismatches at low power ratio and low SNR ($\rho_{\mbox{\scriptsize tot}}=0\textrm{dB}$ and $1\textrm{dB}$) are due to the fact that less power allocated to SW with respect to the noise power results in an increase of the FS error probability for which the union bound becomes useless (i.e. $P_{e,u}$ exceeds $1$). However, the optimum overheads obtained by simulation and by the theoretical expression remain coherent.
 
As illustrated in \figurename{}~\ref{fig:fs_ed_32bits}, above $6\textrm{dB}$, the Monte-Carlo simulations become \emph{extremely costly} due to tiny frame error rates around $10^{-10}$. This demonstrates the practical interest of this theoretical model as the optimal frame structure can be found very fast (minutes compared to days of Monte-Carlo simulations) by evaluating the proposed theoretical bounds, which are close to the Monte-Carlo simulations.

Finally the performance obtained with a practical setup using QPSK modulation combined with the 3GPP 5G NR Downlink Polar channel code~\cite{3gpp382122018} is illustrated in \figurename{}~\ref{fig:fs_polar_32bits}.
The channel code decoder is the CRC-Aided Successive Cancellation List Decoder with list size 32.
This 3GPP Polar code is powerful for short packets and includes a 24-bit CRC that helps the list-based channel decoder. The performance gap between the proposed theoretical bounds and the practical setup partly comes from the sub-optimal decoder (a larger list is required to reach ML performance), and also from the QPSK modulation which departs from the spherical uniform distribution assumption. 
 However the main purpose of this comparison is to underline that the optimal overheads coincide: the proposed theoretical bounds help finding the optimal overheads very fast thus avoiding time-consuming simulations.


\begin{figure}
	\centering
	\includegraphics[width=\linewidth]{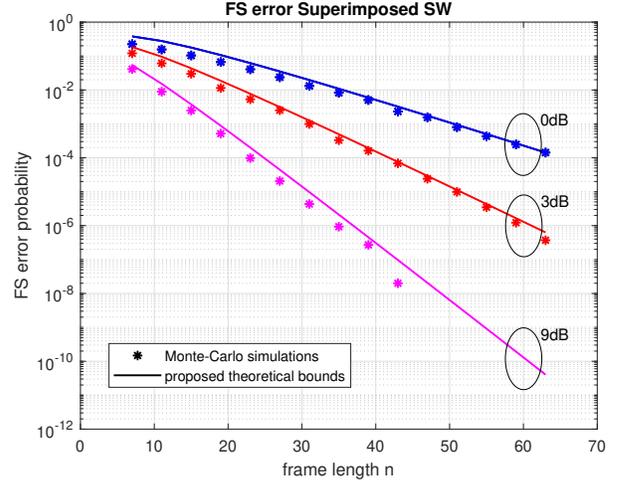}
	\caption{FS error probability and its approximated union bounds for $\rho_{\mbox{\scriptsize tot}}=0\textrm{dB},3\textrm{dB} \textrm{ and } 9\textrm{dB}$ for short frames. Equal power allocation for SW and data ($\alpha = 0.5$).}
	\label{fig:fs_Q_optim_2frame}
\end{figure}
\begin{figure}
	\centering
	\includegraphics[width=\linewidth]{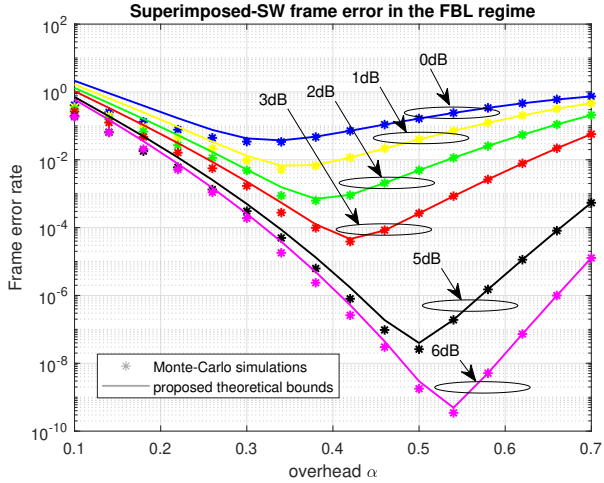}
	\caption{Trade-off between FS and channel-coding. Zadoff-Chu sequences of root 1 are used as SW. Frame length $n=63$ and the number of information bits $k=32$. The total power $\rho_{\mbox{\scriptsize tot}}$ varies from $0\textrm{dB}$ to $6\textrm{dB}$.}
	\label{fig:fs_ed_32bits}
\end{figure}
\begin{figure}
	\centering
    \includegraphics[width=\linewidth]{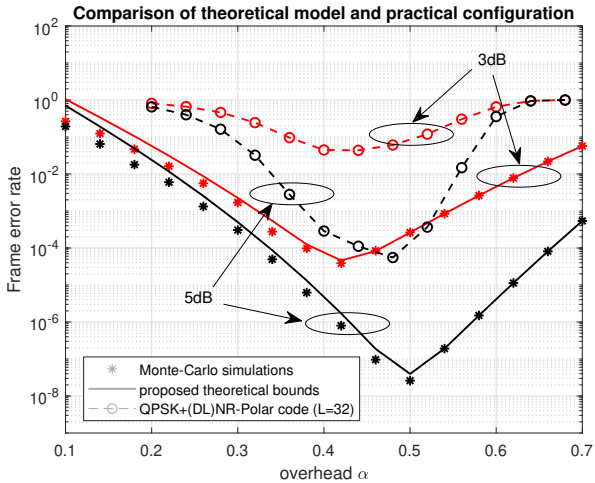}
	\caption{Optimal frame structure comparison with QPSK and 3GPP 5G NR Downlink Polar code. Zadoff-Chu sequences of root 1 are used as SW. Frame length $n=63$ symbols and $k=32$ information bits. The total power $\rho_{\mbox{\scriptsize tot}} = 3\textrm{dB}$ and $5\textrm{dB}$.}
	\label{fig:fs_polar_32bits}
\end{figure}

\section{Conclusion}
\label{sect:conclusion}
In this paper, we have developed an upper bound on the frame synchronization error probability in the context of continuous transmissions over AWGN channel. A synchronization word is assumed to be superimposed to the data symbols instead of being embedded in a frame header. The advantage is that the synchronization word length is as long as the frame itself. The proposed error probability is used to optimize the power overhead between the synchronization word and the data symbols for a given total transmitted average symbol energy using recent results on the finite blocklength coding performances. Comparisons with Monte-Carlo simulations of the theoretical scheme show that our results are tight enough to find the optimal power overhead. Furthermore, numerical evaluations show that the optimal overheads obtained with this theoretical model coincide with the practical 3GPP 5G NR Downlink scheme. The comparison between using a superimposed synchronization word and concatenating it to data symbols is part of ongoing work. Further perspectives also include an extension to mobile fading channels.

\bibliographystyle{IEEEtran}
\bibliography{./isct2018_5pages}

\begin{thebibliography}{10}
\providecommand{\url}[1]{#1}
\csname url@samestyle\endcsname
\providecommand{\newblock}{\relax}
\providecommand{\bibinfo}[2]{#2}
\providecommand{\BIBentrySTDinterwordspacing}{\spaceskip=0pt\relax}
\providecommand{\BIBentryALTinterwordstretchfactor}{4}
\providecommand{\BIBentryALTinterwordspacing}{\spaceskip=\fontdimen2\font plus
\BIBentryALTinterwordstretchfactor\fontdimen3\font minus
  \fontdimen4\font\relax}
\providecommand{\BIBforeignlanguage}[2]{{%
\expandafter\ifx\csname l@#1\endcsname\relax
\typeout{** WARNING: IEEEtran.bst: No hyphenation pattern has been}%
\typeout{** loaded for the language `#1'. Using the pattern for}%
\typeout{** the default language instead.}%
\else
\language=\csname l@#1\endcsname
\fi
#2}}
\providecommand{\BIBdecl}{\relax}
\BIBdecl

\bibitem{ostman2017finite}
J.~{\"O}stman, G.~Durisi, and E.~G. Str{\"o}m, ``Finite-blocklength bounds on
  the maximum coding rate of {Rician} fading channels with applications to
  pilot-assisted transmission,'' in \emph{2017 IEEE 18th International Workshop
  on Signal Processing Advances in Wireless Communications (SPAWC)}.\hskip 1em
  plus 0.5em minus 0.4em\relax IEEE, 2017, pp. 1--5.

\bibitem{3gppfblcode2016}
``Channel coding schemes for {URLLC} scenario,'' R1-1611692 TSG RAN WG1 Meeting
  87, 3GPP, 2016.

\bibitem{durisi2016short}
G.~Durisi, T.~Koch, J.~{\"O}stman, Y.~Polyanskiy, and W.~Yang, ``Short-packet
  communications over multiple-antenna {Rayleigh-fading} channels,'' \emph{IEEE
  Transactions on Communications}, vol.~64, no.~2, pp. 618--629, 2016.

\bibitem{polyanskiy2010channel}
Y.~Polyanskiy, H.~V. Poor, and S.~Verd{\'u}, ``Channel coding rate in the
  finite blocklength regime,'' \emph{IEEE Transactions on Information Theory},
  vol.~56, no.~5, pp. 2307--2359, 2010.

\bibitem{erseghe2016coding}
T.~Erseghe, ``Coding in the finite-blocklength regime: Bounds based on laplace
  integrals and their asymptotic approximations,'' \emph{IEEE Transactions on
  Information Theory}, vol.~62, no.~12, pp. 6854--6883, 2016.

\bibitem{imad2009blind}
R.~Imad, G.~Sicot, and S.~Houcke, ``Blind frame synchronization for error
  correcting codes having a sparse parity check matrix,'' \emph{IEEE
  Transactions on Communications}, vol.~57, no.~6, pp. 1574--1577, 2009.

\bibitem{chiani2007analysis}
M.~Chiani and M.~G. Martini, ``Analysis of optimum frame synchronization based
  on periodically embedded sync words,'' \emph{IEEE Transactions on
  Communications}, vol.~55, no.~11, pp. 2056--2060, 2007.

\bibitem{wang2007novel}
Y.~Wang, J.~Oostveen, A.~Filippi, and S.~Wesemann, ``A novel preamble scheme
  for packet-based {OFDM WLAN},'' in \emph{Wireless Communications and
  Networking Conference, 2007. WCNC 2007. IEEE}.\hskip 1em plus 0.5em minus
  0.4em\relax IEEE, 2007, pp. 1481--1485.

\bibitem{chiani2005practical}
M.~Chiani and M.~G. Martini, ``Practical frame synchronization for data with
  unknown distribution on {AWGN} channels,'' \emph{IEEE Communications
  Letters}, vol.~9, no.~5, pp. 456--458, 2005.

\bibitem{chiani2006sequential}
------, ``On sequential frame synchronization in {AWGN} channels,'' \emph{IEEE
  Transactions on Communications}, vol.~54, no.~2, pp. 339--348, 2006.

\bibitem{suwansantisuk2008frame}
W.~Suwansantisuk, M.~Chiani, and M.~Z. Win, ``Frame synchronization for
  variable-length packets,'' \emph{IEEE Journal on Selected Areas in
  Communications}, 2008.

\bibitem{elzanaty2017frame}
A.~Elzanaty, K.~Koroleva, S.~Gritsutenko, and M.~Chiani, ``Frame
  synchronization for {M-ary} modulation with phase offsets,'' in \emph{2017
  IEEE 17th International Conference on Ubiquitous Wireless Broadband
  (ICUWB)}.\hskip 1em plus 0.5em minus 0.4em\relax IEEE, 2017, pp. 1--6.

\bibitem{bana2018short}
A.-S. Bana, K.~F. Trillingsgaard, P.~Popovski, and E.~de~Carvalho, ``Short
  packet structure for {Ultra-Reliable} {Machine-type Communication}: Tradeoff
  between detection and decoding,'' \emph{International Conference on
  Acoustics, Speech and Signal Processing (ICASSP) 2018}, 2018.

\bibitem{barker1953group}
R.~H. Barker, ``Group synchronization of binary digital systems,'' \emph{in
  Communication Theory}, pp. 273--287, 1953.

\bibitem{massey1972optimum}
J.~Massey, ``Optimum frame synchronization,'' \emph{IEEE Transactions on
  Communications}, vol.~20, no.~2, pp. 115--119, 1972.

\bibitem{lui1987frame}
G.~Lui and H.~Tan, ``Frame synchronization for {Gaussian} channels,''
  \emph{IEEE Transactions on Communications}, vol.~35, no.~8, pp. 818--829,
  1987.

\bibitem{jia2005frame}
H.~Jia and D.~E. Dodds, ``Frame synchronization for {PSAM} in {AWGN} and
  {Rayleigh} fading channels,'' in \emph{2005. Canadian Conference on
  Electrical and Computer Engineering}.\hskip 1em plus 0.5em minus 0.4em\relax
  IEEE, 2005, pp. 44--50.

\bibitem{ling:2017}
F.~Ling, ``{Signal Model and ML Detection of Known Symbol Sequence in Received
  Signal},'' in \emph{{Synchronization in Digital Communication
  Systems}}.\hskip 1em plus 0.5em minus 0.4em\relax Cambridge University Press,
  2017, ch.~3, p.~69.

\bibitem{zepernick2013pseudo}
H.-J. Zepernick and A.~Finger, \emph{Pseudo random signal processing: theory
  and application}.\hskip 1em plus 0.5em minus 0.4em\relax John Wiley \& Sons,
  2013.

\bibitem{3gpp362112017}
\emph{{LTE; E-UTRA;} Physical channels and modulation}, 3GPP Std. 36.211
  v14.6.0, 2018.

\bibitem{3gpp382122018}
\emph{{LTE; E-UTRA;} 5G; NR; Multiplexing and channel coding Rel.15}, 3GPP Std.
  38.212 V15.2.0, 2018.

\end{thebibliography}
\end{document}